\documentclass{article}
\usepackage[margin=4cm]{geometry}
\usepackage[T1]{fontenc}

\usepackage{hyperref}
\usepackage{color}

\usepackage{amssymb}
\usepackage{amsmath,mathtools}
\usepackage{mleftright}
\usepackage{calrsfs}

\usepackage{enumerate}

\newcommand{\integerrange}[2]{[#1\!:\!#2]}
\newcommand{\natrange}[1]{\left[#1\right]}
\newcommand{\Prob}[1]{\operatorname{Pr}\left[#1 \right]}
\newcommand{\semiorderedPattern}[2]{\mathsf{SOP}(#1, #2)}
\newcommand{\subpermutations}[2]{S_{#1,#2}}
\newcommand{\identityperm}{\operatorname{id}}
\newcommand{\lcm}[1]{\operatorname{lcm}\mleft(#1\mright)}

\newcommand{\substr}[3]{#1_{[#2..#3]}}
\newcommand{\kderang}[2]{\mathsf{Derang}(#1,#2)}
\newcommand{\wasteclass}[2]{\mathsf{Wst}(#1,#2)}

\usepackage{amsthm}

\newtheorem{theorem}{Theorem}
\newtheorem{proposition}[theorem]{Proposition}

\newtheorem{example}[theorem]{Example}
\newtheorem{remark}[theorem]{Remark}

\usepackage{array}
\definecolor{Gray}{gray}{0.85}
\newcolumntype{a}{>{\columncolor{Gray}}c}
\newcolumntype{b}{>{\columncolor{white}}c}
\usepackage{colortbl}
\usepackage{multirow}

\usepackage{everypage}
\def\PageTopMargin{1cm}
\def\PageLeftMargin{1cm}
\newcommand\atxy[3]{%
  \AddThispageHook{\hbox{\smash{\hspace*{\dimexpr-\PageLeftMargin-\hoffset+#1\relax}%
    \raisebox{\dimexpr\PageTopMargin+\voffset-#2\relax}{\textcolor{gray}{#3}}}}}}
\usepackage{orcidlink}

\begin{document}

\title{SAT-Based Search for \\ Minwise Independent Families}%

\author{Enrico Iurlano\,\orcidlink{0000-0001-7528-0834}\, Günther~R. Raidl\,\orcidlink{0000-0002-3293-177X}\,}

\date{December 16, 2024}

\maketitle
\begingroup%
\def\thefootnote{\fnsymbol{footnote}}\footnote[0]{Algorithms and Complexity Group, TU Wien, Favoritenstraße 9--11/192-01, Vienna, 1040, Austria \par Email:
  \texttt{\{eiurlano|raidl\}@ac.tuwien.ac.at}}%
\endgroup
\begin{abstract}
  Proposed for rapid document similarity estimation in web search engines, the celebrated property of minwise independence imposes highly symmetric constraints on a family $\mathcal{F}$ of permutations of $\{1,\ldots, n\}$:
  The property is fulfilled by $\mathcal{F}$ if for each $j\in \{1,\ldots,n\}$, any cardinality-$j$ subset $X\subseteq \{1,\ldots,n\}$, and any fixed element $x^\ast\in X$, it occurs with probability $1/j$ that a randomly drawn permutation $\pi$ from $\mathcal{F}$ satisfies $\pi(x^\ast)=\min \{\pi(x) : x\in X\}$.
  The central interest is to find a family with fewest possible members meeting the stated constraints.
  We provide a framework that, firstly, is realized as a pure SAT model and, secondly, generalizes a heuristic of Mathon and van Trung to the search of these families.
  Originally, the latter enforces an underlying group-theoretic decomposition to achieve a significant speed-up for the computer-aided search of structures which can be identified with so-called rankwise independent families.
  We observe that this approach is suitable to find provenly optimal new representatives of minwise independent families while yielding a decisive speed-up, too.
  As the problem has a naive search space of size at least $(n!)^n$, we also carefully address symmetry breaking.
  Finally, we add a bijective proof for a problem encountered by Bargachev when deriving a lower bound on the number of members in a minimal rankwise independent family.
\end{abstract}
\begin{center}
\atxy{3.75cm}{25.45cm}{\parbox{18.5cm}{
    \footnotesize \texttt{Copyright \textcopyright{} 2024, the authors.}
    \texttt{This document is a preprint.}
  }}
\end{center}

\section{Introduction}
The concept of ($k$-restricted) minwise independence was introduced by Broder et al.~\cite{broder2000minwise} in the context of rapid similarity estimation within large document collections on the World Wide Web.
The result of their considerations is ultimately the \texttt{MinHash} algorithm, which is widely used by practitioners in the field of big data (see~\cite{zamora2016hashing} as an example of thousands of other works that build on~\cite{broder2000minwise}).
It can accurately approximate the pairwise Jaccard similarity between high-dimensional data by precomputing in linear time the so-called \emph{sketches}~\cite{broder2000minwise} associated to the data.
Let us take a closer look at the most important observation used for this:
For estimating the Jaccard similarity $r(A,B):=|A\cap B|/|A\cup B|$ between two sets $A,B\subseteq\natrange{n}:=\{1, \ldots,n\}$, it was figured out~\cite{broder2000minwise} that, with $E_{\pi;A,B}$ describing the scenario of coincidence
\begin{equation*}
  \min\{\pi(a): a\in A\} = \min\{\pi(b): b\in B\},
\end{equation*}
we have
\begin{equation}
  \Prob{E_{\pi;A,B} \,|\,\pi\text{ drawn uniformly at random from }S_n}= r(A, B),\label{eq:jaccardsimilarity-approximation-randomly-drawing}
\end{equation}
where $S_n$ stands for the symmetric group consisting of all permutations, i.e., bijective functions from $\natrange{n}$ onto $\natrange{n}$.
Hence, the Jaccard similarity can be characterized by the (conditional) probability stated in~\eqref{eq:jaccardsimilarity-approximation-randomly-drawing}.
However, the difficulty of accurately approximating this similarity measure $r$ via random sampling lies in the fact that, for larger values of $n$, it is hard to draw a permutation from $S_n$---even approximately---uniformly at random.
The reason is the loss of accuracy of pseudorandom number generators in large numeric scopes.
Therefore, Broder et al.~\cite{broder2000minwise} studied families of permutations, which they called \emph{minwise independent}, having less members than $S_n$ and yet ensuring the identity in~\eqref{eq:jaccardsimilarity-approximation-randomly-drawing}.

Minwise independence later became a useful concept and tool for many other applications, including, but not limited to, the detection of spam~\cite{li2011theory}, machine learning (dimensionality reduction)~\cite{zamora2016hashing}, and derandomization as well as computational geometry~\cite{mulmuley1994computational}.

\medskip

We now formally introduce the most basic concepts and properties of minwise independence.
Following~\cite{broder2000minwise}, a finite non-empty family $\mathcal{F}=(\pi_1,\ldots, \pi_d)$ of permutations $\pi_i\in S_n$, $i=1,\ldots,d$, together with a probability distribution on the members of $\mathcal{F}$, $p=(p_i)_{i=1}^d\in[0,1]^{d}$, is \emph{$k$-restricted minwise independent} if for each $j\in\natrange{k}$ for each $X\subseteq \natrange{n}$ with $|X|=j$, and each $x^\ast$ in $X$,
\begin{equation}
  \Prob{\pi(x^\ast) = \min_{x\in X} \pi(x) ~\vert~ \pi\text{ drawn from }\mathcal{F}\text{ randomly according to }p} = \frac{1}{j}.\label{eq:k-restriced-minwise-with-probability-distribution-definition}
\end{equation}
If not stated otherwise, in the following, we will always assume a uniform probability distribution on the members of a family, i.e., $p = (1/d)_{i=1}^d$.
Let us denote by $\subpermutations{n}{k}$ the set of injective functions with domain $\natrange{k}$ and codomain $\natrange{n}$, henceforth called \emph{$k$-subpermutations}, and let us introduce $\semiorderedPattern{n}{j} := \{(s_1,\ldots,s_j)\in \natrange{n}^j: s_2<s_3<\ldots < s_j \wedge s_1 \not\in\{s_2,\ldots,s_j\}\}$ with $1\leq j\leq n$, which we call the set of \emph{semiordered patterns}.
We can state the equality~\eqref{eq:k-restriced-minwise-with-probability-distribution-definition} equivalently as
\begin{equation}
  \lvert\{i\in \natrange{d} : \pi_i(s_1) = \min\{\pi_i(s_1), \pi_i(s_2),\ldots,\pi_i(s_{j})\}\}\rvert = \frac{d}{j},\label{eq:minwise-independence-definition-via-cardinality-counting}
\end{equation}
for each $j\in\natrange{k}$ and each $(s_1,s_2,\ldots,s_j)\in\semiorderedPattern{n}{j}$, implying that $d$ must be a multiple of $\lcm{\natrange{k}}$, the least common multiple of the numbers $1,2,\ldots, k$.
We call a $k$-restricted minwise independent family $\mathcal{F}=(\pi_1,\ldots,\pi_d)\subseteq S_n$ \emph{optimal} if for all $\tilde{d}<d$ there is no $k$-restricted minwise independent family $\widetilde{\mathcal{F}}=(\pi_1,\ldots,\pi_{\tilde{d}})\subseteq S_n$.

\begin{example}
  The family $\mathcal{F}:= (\pi_1,\ldots,\pi_6)\subseteq S_4$ given by
  \begin{equation*}
    ((2,3,1,4), (1,4,2,3), (4,1,2,3), (2,3,4,1), (2,1,4,3), (4,3,2,1))
  \end{equation*}
  is $3$-restricted minwise independent---the permutations are here given in tuple notation.
  Notice that for $(2,1,4)\in\semiorderedPattern{4}{3}$, we have $2=|\{\pi_3, \pi_5\}|=6/3$ and $\pi_3(2) = 1 = \min\{1,4,3\} = \min \{\pi_3(2), \pi_3(1) ,\pi_3(4)\}$ as well as $\pi_5(2)=1=\min\{1,2,3\}=\min\{\pi_5(2),\pi_5(1), \pi_5(4)\}$.
  A respective observation can be made for all $12$ elements in $\semiorderedPattern{4}{3}$, all $12$ elements in $\semiorderedPattern{4}{2}$, and all four elements in $\semiorderedPattern{4}{1}$.
\end{example}

\begin{remark}\label{rem:identifiability-perfect-sequence-covering-arrays-rankwise-independent-families-and-further-literature}
  The more constrained class of \emph{$k$-rankwise independent families}~\cite{tarui2003nearly} is defined by the requirement that for each subpermutation $\sigma\in\subpermutations{n}{k}$
  \begin{equation}
    \lvert\{i\in \natrange{d} : \pi_i(\sigma(1)) < \ldots < \pi_i(\sigma(k))\}\rvert = d/k!. \label{eq:rankwise-indep-formulation}
  \end{equation}
  As recognized in~\cite{iurlano2023growth}, the latter could be equivalently described in the language of \emph{perfect sequence covering arrays}~\cite{yuster2020perfect} for which several construction/search approaches have been proposed~\cite{levenshtein1992perfect,mathon1999directed,yuster2020perfect,na2023group,gentle2023perfect,iurlano2023growth,gentle2023polynomial} and which form alternatives to those established for rankwise independence~\cite{tarui2003nearly,bargachev2004improved}.
\end{remark}
\begin{remark}\label{rem:three-restricted-minwise-independence-is-three-rankwise-independence}
  It is an observation in~\cite{tarui2003nearly} that the validity of~\eqref{eq:minwise-independence-definition-via-cardinality-counting} for $j\leq 3$ is equivalent to the validity of~\eqref{eq:rankwise-indep-formulation} with $k=3$, and consequently $3$-restricted minwise independence is precisely $3$-rankwise independence.
\end{remark}
The next theorem resembles the tightest so-far known asymptotic bounds for general $n$ and $k$.
For integers $n$, we denote by $n!$ the factorial and by $!n := (n!)\cdot \sum_{i=0}^n {(-1)}^{i}/(i!)$ the \emph{subfactorial}.
\begin{theorem}\label{thm:asymptotic-overview}
  Let $\mathcal{F}\subseteq S_n$ be $k$-restricted minwise independent and $\mathcal{G}\subseteq S_n$ be $k$-rankwise independent.
  Then, the following estimates apply:
  \begin{enumerate}[(i)]
    \item $|\mathcal{F}| \geq \max \{n, \lcm{\natrange{k}}\}$~\cite{itoh2000permutations,bargachev2004improved}. \label{ite:lower-bound-minwise-independent-families}
    \item $|\mathcal{F}|\leq n^{(1+(1/\ln n))k}\lcm{\natrange{k-1}}$~\cite{itoh2000permutations}. \label{ite:upper-bound-minwise-independent-families}
    \item $|\mathcal{G}|\geq \sum_{i=0}^{\lfloor k/2\rfloor} !i \binom{n}{i} =: E(n,k)$ for even $k$, otherwise, when $k$ is odd, $|\mathcal{G}|\geq E(n,k)\allowbreak+\,!\lceil k/2\rceil\binom{n-1}{\lfloor k/2\rfloor}$~\cite{bargachev2004improved}.\label{ite:bargachev-lower-bound}
    \item There is some constant $C>0$ (independent of $n$ and $k$), for which $|\mathcal{G}|\leq (Cn)^{Ck}$~\cite{kuperberg2017probabilistic}.
          There is some constant $D>0$ (independent of $n$ and $k$), for which $|\mathcal{G}|\leq (Dn)^{35k}$~\cite{harvey2024explicit}.\label{ite:upper-bound-rankwise-independent-families}
  \end{enumerate}
\end{theorem}
\begin{remark}
  In~\cite[p.~141]{itoh2000permutations} a simple transformation is given, which turns a $k$-restricted \mbox{min-/}rankwise independent family $\mathcal{F} \subseteq S_n$ into a $k$-restricted min-/rankwise independent family $\widetilde{\mathcal{F}}\subseteq S_{\tilde{n}}$ with $|\mathcal{F}|=|\widetilde{\mathcal{F}}|$, for any $\tilde{n}\leq n$.
  By the contrapositive, non-existence of $k$-restricted minwise independent families $\widetilde{\mathcal{F}}\subseteq S_{\tilde{n}}$ of $d$ members implies non-existence for such families $\mathcal{F}\subseteq S_n$ with larger $n\geq \tilde{n}$, provided $|\mathcal{F}|=d$.
\end{remark}
Finally, let us recall some elementary and here crucial concepts from group theory: For any group $A$ with binary group operation $\circ : A\times A \to A$, neutral element $e\in A$, and element inversion $(\cdot)^{-1}: A\to A$, for each subgroup $G$ of $A$ we can form a so-called left-coset by picking an arbitrary $a\in A$ and forming $aG := \{a\circ g : g\in G\}$.
Similarly, right-cosets $Ga := \{g\circ a: g \in G\}$ are formable for $a\in A$.
We will exclusively focus on subgroups of the symmetric group $S_n$ with the composition of functions as operation (see also Remark~\ref{rem:clarification-notation-composition-of-functions-operator}).

For minwise independence, the following result seems to be the first in the literature which (after casting the original assertion to the present group theoretical setting) proves the existence of a minwise independent family by a particular decomposition into \emph{right-cosets}.
\begin{theorem}[\cite{bargachev2006some}]\label{thm:bargachev-first-right-coset-approach-for-minwise-independence}
  Let $k\geq 3$ be odd and $\mathcal{F}=(\theta_1,\ldots,\theta_d)\subseteq S_n$ be $k$-restricted minwise independent.
  With $G=\{\gamma_1,\gamma_2\} := \{\identityperm, \sigma\}$ denoting the subgroup of $S_n$ containing the identity permutation and the order reversing permutation $\sigma: \natrange{n} \to \natrange{n}$, $i\mapsto n+1-i$, we have that $(\gamma_\ell \circ \theta_m : \ell\in \natrange{d}, m\in\natrange{2})$ is a $(k+1)$-restricted minwise independent family counting $2d$ members.
\end{theorem}
We also emphasize that a connection of minwise independence to group theory has been studied within a more restrictive setting in~\cite{cameron2007minwise}.

\medskip

The goal of this paper is to provide a SAT framework enabling the study---for reasonably small scales of $n$ and $k$---the existence and nature of $k$-restricted minwise independent families $\mathcal{F}\subseteq S_n$ with a pre-specified number of members $d$.
Analogously to~\cite{na2023group,gentle2023perfect}, a particular emphasis will be set on the question whether (near-) optimal such families can be ``spanned'' by suitable---possibly large---subgroups of $S_n$ and to quantify a potential speed-up in terms of solving time.

\section{A SAT approach}\label{sec:a-sat-approach}

The main idea shares the natural approach from~\cite{banbara2012generating} (in the setting of Answer Set Programming) where permutations are represented via their order theoretic incidence structure.
In fact, if $\pi\in S_n$ is a permutation, it is disambiguously determined by  full knowledge on $X^\pi = (x^\pi_{i,j})_{i,j=1}^n\in\{0,1\}^{n\times n}$ capturing by $x^\pi_{i,j}$ the information if $\pi(i)<\pi(j)$ ($1$ in the affirmative case, otherwise $0$).
In the following it will become clear that considering just the strict upper-diagonal entries of $X^\pi$ is sufficient (see later asymmetry).
Note that $X^\pi$ can be seen as incidence matrix of a strict total order on $\natrange{n}$ representing the permutation.
We can in fact recover the $j$-th entry of the permutation via $\pi(j) = 1+\sum_{i=1}^n x^\pi_{i,j}$.
All solutions for $x_{i,j}^\pi$ subject to the following constraints (irreflexivity, asymmetry, and transitivity) determine hence a permutation:
\begin{align}
  \neg x^{\pi_\ell}_{i,i}                                                      & ~~~\forall\ell\in\natrange{d}~\forall i\in\natrange{n},\label{eq:irreflexitivity}                                                                           \\
  \neg x^{\pi_\ell}_{i,j} \vee \neg x^{\pi_\ell}_{j,i}                         & ~~~\forall\ell\in\natrange{d}~\forall i\in\natrange{n}~\forall j\in\natrange{n},\label{eq:asymmetry}                                                        \\
  \neg x^{\pi_\ell}_{i,j} \vee \neg x^{\pi_\ell}_{j,h} \vee x^{\pi_\ell}_{i,h} & ~~~\forall \ell\in\natrange{d}~\forall i\in\natrange{n}~\forall j\in\integerrange{i+1}{n}~\forall h\in\natrange{n}\setminus \{i,j\}.\label{eq:transitivity}
\end{align}
These constraints are already stated in conjunctive normal form (CNF), and for two integers $p$ and $p$, we denote by $\integerrange{p}{q} := \{p,p+1,\ldots, q\}$.

To model property~\eqref{eq:minwise-independence-definition-via-cardinality-counting} in the definition of minwise independence, we need to impose specific cardinality constraints which are stated in~\eqref{eq:sopeq-fulfillment-without-j-at-most-three}---note that $j\leq 3$ is intentionally excluded due to a more favorable replacement in~\eqref{eq:explicit-three-subpermutation-fulfillment}:
\begin{align}
  \sum_{\ell=1}^d \left[\bigwedge_{h=2}^j x^{\pi_\ell}_{s_1,s_h}\right] \leq d/j                & ~~~\forall j \in \integerrange{4}{k} ~\forall s\in \semiorderedPattern{n}{j},\label{eq:sopeq-fulfillment-without-j-at-most-three} \\
  \sum_{\ell=1}^d \left[\bigwedge_{h=2}^3 x^{\pi_\ell}_{\sigma(h-1),\sigma(h)}\right] \leq d/3! & ~~~ \forall \sigma\in \subpermutations{n}{3} \label{eq:explicit-three-subpermutation-fulfillment}.
\end{align}
In fact, we avoid the higher upper bound $d/3$ (which in~\eqref{eq:sopeq-fulfillment-without-j-at-most-three} would result for $j=3$) and also can relinquish the case $j=2$---this is justified by the aforementioned particularity that $3$-restricted minwise independence is precisely $3$-rankwise independence.

Note that upper bounds instead of equalities appear in~\eqref{eq:sopeq-fulfillment-without-j-at-most-three}--\eqref{eq:explicit-three-subpermutation-fulfillment} which provide a considerable simplification in view of a conversion to CNF.
These seemingly weaker conditions are however equivalent to the ones with equalities: Any not attained upper bound in~\eqref{eq:sopeq-fulfillment-without-j-at-most-three}--\eqref{eq:explicit-three-subpermutation-fulfillment} would imply a strict excess of the upper bound for some other $s'\in\semiorderedPattern{n}{j}$ respectively $\sigma'\in\subpermutations{n}{3}$---infeasibility would be a consequence; in fact, the number of patterns/subpermutation conditions met by a family $\mathcal{F}$ is fixed.
How the cardinality constraints are technically translated to CNF is deferred to Section~\ref{sec:computational-experiments}.
\begin{remark}
  If the interest is in modeling $k$-rankwise independence, we can simply replace~\eqref{eq:sopeq-fulfillment-without-j-at-most-three}--\eqref{eq:explicit-three-subpermutation-fulfillment} with $\sum_{\ell=1}^d \left[\bigwedge_{h=2}^k x^{\pi_\ell}_{\sigma(h-1),\sigma(h)}\right] \leq d/k!~~\forall \sigma\in \subpermutations{n}{k}$.
\end{remark}

For symmetry breaking, for a permutation $\pi$, we can consider the binary string resulting from $X^\pi = (x^{\pi}_{i,j})_{i,j=1}^n$ after the following procedure: All side-diagonals of $X^\pi$ lying above the main diagonal are concatenated as string denoted as $z_{\textrm{cat}}[X^\pi] := (X^\pi_{i+1,i+j}: j=2,\ldots n, i=0,\ldots, n-j)$.
It is then natural, without loss of generality, to enforce an ordering on $x^{\pi_\ell}$, $\ell = 1, \ldots, d$, which is based on the lexicographical ordering of their associated strings $z_{\textrm{cat}}[X^{\pi_{\ell}}]$.
Thus, when $\preceq$ denotes the lexicographical ordering relation on binary strings (e.g., $\texttt{0100}\preceq \texttt{0110}$), we require
\begin{align}
  z_{\textrm{cat}}[X^{\pi_{m+1}}] \preceq z_{\textrm{cat}}[X^{\pi_{m}}] & ~~~ \forall \ell \in\natrange{d}.\label{eq:lex-ordering-imposed-on-permutations-based-on-their-string-pearl-up}
\end{align}
The identity permutation, whose upper diagonal part of the incidence matrix consists of exclusively $1$-entries, is hence the minimum element with respect to this ordering.
For $\pi\in S_n$ the string $z_{\textrm{cat}}[X^\pi]$ is of length $n(n-1)/2$.
Therefore, we might prefer the following weakened constraint depending on a chooseable ``accuracy'' parameter $H\in\integerrange{0}{n(n-1)/2}$; for a string $s$ we denote the substring consisting of the first $H$ entries of $s$ by $\substr{s}{1}{H}$:
\begin{align}
  \substr{z_{\textrm{cat}}[X^{\pi_{m+1}}]}{1}{H} \preceq \substr{z_{\textrm{cat}}[X^{\pi_{m}}]}{1}{H} & ~~~ \forall \ell \in \natrange{d}.\label{eq:lexicographically-ordered-subblocks}
\end{align}
In~\cite[entry~A036604]{oeis2024online} one can look up $\mathsf{A}(n)$, the number of  pairwise comparisons (dynamically chooseable) needed to uniquely determine a strict total order on $\natrange{n}$; see also the follow-up paper~\cite{peczarski2004new}.
In the following we regard the choice $H:=\mathsf{A}(n)$ as meaningful guideline but could also change this value for $H$, if a different trade-off between symmetry breaking and lower model complexity is preferred.
Finally, in this setting we can add without loss of generality the following constraint.
As in~\cite[Lemma~4.2]{na2023group}, we preserve minwise independence by composing each permutation in the family with $\pi_1^{-1}$ and afterwards bringing the transformed in lexicographic order:
\begin{align}
  \pi_1=\operatorname{id}\text{, i.e., unit clauses~~~}x^{\pi_1}_{i,j} & ~~~ \forall i\in\natrange{n}~\forall j\in\integerrange{i+1}{n}.
\end{align}

We now transfer to our setting a heuristic approach relying on group theory due to Mathon and van Trung~\cite{mathon1999directed}.
Its successful applicability to the search for perfect sequence covering arrays of higher ``coverage-multiplicity'' has been recently shown~\cite{na2023group}.
In the spirit of these two previous works, we will assume that our families $(\pi_1,\ldots, \pi_d)$ with $\pi_i\in S_n$ and where $\lcm{\natrange{k}}$ is a divisor of $d$ have the following underlying structure:
There is a subgroup $G=\{\gamma_1,\ldots, \gamma_q\}$ of $S_n$ whose order $q = |G|$ is a divisor of $d$ and for which there are ``offsets'' $\theta_1,\ldots,\theta_{d/|G|}\in S_n$ such that
\begin{equation}
  \mathcal{F} = (\pi_1,\ldots,\pi_d)=(\theta_{\ell}\circ \gamma_{m})_{(\ell,m)\in\natrange{d/|G|}\times\natrange{|G|}}.\label{eq:family-composed-of-leftcosets}
\end{equation}
Condition~\eqref{eq:family-composed-of-leftcosets} models the coincidence of $(\pi_1,\ldots,\pi_d)$ with a union of $d/|G|$ not necessarily different left-cosets of a subgroup $G$ of $S_n$.
Hence, having fixed a non-trivial subgroup $G$ of $S_n$, the number of decision variables is reduced, as every choice for an ``offset'' $\theta_i$ automatically fixes a larger set of size $|G|$ of permutations in $\mathcal{F}$.
Note that we will always assume that per default $\gamma_1 := \operatorname{id}$ in our enumeration of $G$.

A model for this heuristic approach in the above given setting consists of the following parts:
Decision variables $X^{\theta_\ell}$, $\ell=1,\ldots, d/|G|$, with permutation constraints as in~\eqref{eq:irreflexitivity}--\eqref{eq:transitivity}, cardinality constraints
\begin{align}
  \sum_{\ell=1}^{d/|G|} \sum_{m=1}^{|G|}  \left[\bigwedge_{h=2}^j x^{\theta_\ell}_{\gamma_m(s_1),\gamma_m(s_h)}\right] \leq d/j                & ~~~\forall j\in\integerrange{4}{k} ~\forall s\in \semiorderedPattern{n}{j},\label{eq:groupversion-sopeq-fulfillment-without-j-at-most-three} \\
  \sum_{\ell=1}^{d/|G|} \sum_{m=1}^{|G|}  \left[\bigwedge_{h=2}^3 x^{\theta_\ell}_{\gamma_m(\sigma(h-1)),\gamma_m(\sigma(h))}\right] \leq d/3! & ~~~ \forall \sigma\in\subpermutations{n}{3}, \label{eq:groupversion-explicit-three-subpermutation-fulfillment}
\end{align}
and lexicographical symmetry breaking, compare~\eqref{eq:lexicographically-ordered-subblocks},
\begin{align}
  \substr{z_{\textrm{cat}}[\theta_{\ell+1}]}{1}{H} \preceq \substr{z_{\textrm{cat}}[\theta_{\ell}]}{1}{H} & ~~~ \forall \ell \in\natrange{d/|G|-1}.
\end{align}

Fixing the trivial subgroup consisting just of the identity permutation $\{\gamma_1\}=\{\identityperm\}$ of $S_n$, we recover precisely the original, non-heuristic model~\eqref{eq:irreflexitivity}--\eqref{eq:explicit-three-subpermutation-fulfillment},~\eqref{eq:lexicographically-ordered-subblocks}.

\begin{remark}\label{rem:clarification-notation-composition-of-functions-operator}
  For clarity, we emphasize that our notation $\alpha\circ\beta$ specifies the function resulting from applying function $\alpha$ to the output of function $\beta$, i.e., $\alpha(\beta(\cdot))$---in contrast to the notation of~\cite{na2023group}.
\end{remark}

In~\cite{na2023group}, also right-cosets are examined, where the postulated equality~\eqref{eq:family-composed-of-leftcosets} translates to
\begin{equation}
  \mathcal{F} = (\pi_1,\ldots,\pi_d)=(\gamma_{m}\circ \theta_{\ell})_{(\ell,m)\in\natrange{d/|G|}\times\natrange{|G|}}.
  \label{eq:family-composed-of-rightcosets}
\end{equation}
In the framework of SAT, this needs a more elaborated modeling approach:
Although we notice that $\gamma_m\circ\theta_\ell={(\theta_\ell^{-1}\circ \gamma_m^{-1})}^{-1}$, leading to the idea to use $\theta_\ell^{-1}$ directly as new decision variable, there is then no direct way to recover via an encoding in CNF the incidence matrix of the inverse of a permutation (where the latter is itself specified by its incidence matrix).
In fact we now face the difficulty to let a symbolic permutation specified by decision variables take action on the domain of a given collection of parameter permutations and return the output of the latter.
This considerably differs from the left-coset approach where the need is to access the incidence matrix of a symbolic permutation in entries pre-permutated by concretely given parameter permutations.

We circumvent this issue by modeling the offsets $\theta_m$ as permutation matrices  $T^{\theta_\ell}=(t_{ij}^{\theta_\ell})_{i,j=1}^n\in\{0,1\}^{n\times n}$ whose bi-stochasticity we enforce via
\begin{align}
  \bigvee_{j=1}^n t_{ij}^{\theta_\ell}                        & ~~~ \forall \ell\in\natrange{d/q}~\forall i\in\natrange{n},                                           \\
  \neg t_{i,j}^{\theta_\ell} \vee \neg t_{i, r}^{\theta_\ell} & ~~~ \forall \ell\in\natrange{d/q}~\forall i\in\natrange{n}~\forall \{j,r\}\in\binom{\natrange{n}}{2}, \\
  \neg t_{i,j}^{\theta_\ell} \vee \neg t_{r, j}^{\theta_\ell} & ~~~ \forall \ell\in\natrange{d/q}~\forall j\in\natrange{n}~\forall \{i,r\}\in\binom{\natrange{n}}{2}.
\end{align}

The following constraint~\eqref{eq:recover-incidence-from-permutation-matrix} stores the incidence matrix of  $\gamma_m\circ\theta_\ell$ for any fixed $\gamma_m$ and a symbolic $\theta_\ell$ inside the fresh variables $x_{i,r}^{\gamma_m\circ\theta_\ell}$, $(i,r)\in \natrange{n}^2$.
The formulation depends on the permutation matrix $T^{\theta_\ell}$ of $\theta_\ell$:
\begin{align}
  x_{i,r}^{\gamma_m\circ\theta_\ell}\leftrightarrow\left( (t_{\gamma_m(r), j}^{\theta_\ell})_{j=1}^{n}\preceq (t_{\gamma_m(i), j}^{\theta_\ell})_{j=1}^{n}\right) & ~~~\forall \ell\in\natrange{d/q}\forall m\in\natrange{q}\forall (i,r)\in\natrange{n}^2.\label{eq:recover-incidence-from-permutation-matrix}
\end{align}
Regarding~\eqref{eq:recover-incidence-from-permutation-matrix}, note that firstly, row $i$ is a lexicographical successor of row $r$ of the permutation matrix of an arbitrary permutation $\pi$ iff $\pi(i) < \pi(r)$; secondly, the permutation matrix of $\gamma_m\circ\theta_\ell$ can be recovered from $T^{\theta_\ell}$ by applying the permutation $\gamma_m$ to its rows, as illustrated in the following example.
\begin{example}\label{exa:example-perm-incidence-comparison-also-after-inner-application-of-perm}
  Let $(\theta(i))_{i=1}^4=(4,1,3,2)$, $(\gamma(i))_{i=1}^4=(4,2,3,1)$.
  Then, $(\theta\circ\gamma(i))_{i=1}^4 = (2,1,3,4)$.
  The permutation matrix, respectively incidence matrix, of $\theta$ is
  \begin{equation*}
    \setlength\arraycolsep{2pt}
    T=\begin{pmatrix}
      0 & 0 & 0 & 1 \\
      1 & 0 & 0 & 0 \\
      0 & 0 & 1 & 0 \\
      0 & 1 & 0 & 0
    \end{pmatrix},\text{ respectively }I=\begin{pmatrix}
      0 & 0 & 0 & 0 \\
      1 & 0 & 1 & 1 \\
      1 & 0 & 0 & 0 \\
      1 & 0 & 1 & 0
    \end{pmatrix}.
  \end{equation*}
  The permutation matrix, respectively incidence matrix, of $\theta\circ\gamma$ is
  \begin{equation*}
    \setlength\arraycolsep{2pt}
    T'=\begin{pmatrix}
      0 & 1 & 0 & 0 \\
      1 & 0 & 0 & 0 \\
      0 & 0 & 1 & 0 \\
      0 & 0 & 0 & 1
    \end{pmatrix},\text{ respectively }I'=\begin{pmatrix}
      0 & 0 & 1 & 1 \\
      1 & 0 & 1 & 1 \\
      0 & 0 & 0 & 1 \\
      0 & 0 & 0 & 0
    \end{pmatrix}.
  \end{equation*}
\end{example}
The main task to be addressed here is accessing the truth value of $x_{i,r}^{\gamma_m \circ \theta_\ell}$, storing as Boolean whether the $i$-th and the $r$-th row are in ascending lexicographic ordering.
A case distinction shows that imposing the subsequent lexicographical ordering on the tuples suitably linking $x^{\gamma_m \circ \theta_\ell}$ to $T^{\theta_\ell}$-entries precisely models such a behavior for $x_{i,r}^{\gamma_m \circ \theta_\ell}$:
\begin{align}
  \text{\eqref{eq:right-coset-first-matrix-lex-comparison}--\eqref{eq:right-coset-second-matrix-lex-comparison} are satisfied } ~~                  & \forall\ell \in \natrange{d/q} \forall m\in\natrange{q}\forall (i,r)\in\natrange{n}^2;                                                                                                 \\
  (x_{i,r}^{\gamma_m \circ \theta_\ell}, \hphantom{\neg}t_{\gamma_m(r), 1}^{\theta_\ell},~\ldots~, \hphantom{\neg}t_{\gamma_m(r), n}^{\theta_\ell}) & \preceq (1,\hspace{8.2mm} \hphantom{\neg}t_{\gamma_m(i), 1}^{\theta_\ell},~\ldots~, \hphantom{\neg}t_{\gamma_m(i), n}^{\theta_\ell})\label{eq:right-coset-first-matrix-lex-comparison}
  \\
  (0, \neg t_{\gamma_m(r), 1}^{\theta_\ell},~\ldots~,\neg t_{\gamma_m(r), n}^{\theta_\ell})                                                         & \preceq
  (x_{i,r}^{\gamma_m \circ \theta_\ell},\neg t_{\gamma_m(i), 1}^{\theta_\ell},~\ldots~,\neg t_{\gamma_m(i), n}^{\theta_\ell})\label{eq:right-coset-second-matrix-lex-comparison}
\end{align}
The implementation of $\preceq$ in SAT is deferred to Section~\ref{sec:computational-experiments}.
What is left, are the analogous cardinality constraints, now statable as
\begin{align}
  \sum_{\ell=1}^{d/|G|} \sum_{m=1}^{|G|}  \left[\bigwedge_{h=2}^j x^{\gamma_m\circ \theta_\ell}_{s_1,s_h}\right] \leq \hspace{1.28mm} d/j & ~~~\forall j\in\integerrange{4}{k} ~\forall s\in \semiorderedPattern{n}{j},\label{eq:groupversion-sopeq-fulfillment-without-j-at-most-three-right-coset} \\
  \sum_{\ell=1}^{d/|G|} \sum_{m=1}^{|G|}  \left[\bigwedge_{h=2}^3 x^{\gamma_m\circ \theta_\ell}_{\sigma(h-1),\sigma(h)}\right] \leq d/3!  & ~~~ \forall \sigma\in \subpermutations{n}{3}, \label{eq:groupversion-explicit-three-subpermutation-fulfillment-right-coset}
\end{align}
where without loss of generality we can require
\begin{align}
  z_{\textrm{cat}}[X^{\gamma_1\circ \theta_\ell}]_{[1:H]} \preceq z_{\textrm{cat}}[X^{\gamma_1\circ \theta_{\ell+1}}]_{[1:H]} & ~~~ \forall \ell \in \natrange{d/q-1}.\label{eq:lexicographically-ordered}
\end{align}

\medskip

The question is now, which subgroups are promising choices, and how many of them should be tested.
Potentially, for each $q$ dividing $n!$ and $d$, there can exist a subgroup of $S_n$ of order $q$ being fruitful for the coset approach.
In~\cite{na2023group} it is shown for their problem setting that for the left-coset approach, a single representative per \emph{conjugacy class} $\operatorname{Cl}(G) := \{\psi G\psi^{-1}:\psi\in S_n\}$ is sufficient to be considered, as it reflects the behavior of all groups being representatives of this class.
For the right-coset approach, all subgroups of a given order have to be examined, where, however, at least the first offset $\theta_1$ can without loss of generality~\cite{na2023group} be assumed as the identity permutation.
We incorporate these insights into our experimental setup too, for shrinking the number of inspected groups, respectively tightening the search space.
Note that for simplicity and for having a richer pool of instances, we disregard the omission of non-maximal subgroups which is additionally implemented in~\cite{na2023group}.

\section{Computational experiments}\label{sec:computational-experiments}

We use \texttt{Julia} in version 1.11.1 for generating strict total order constraints.
The cardinality constraints are obtained by calling the library \texttt{PySAT} in version 0.1.7.\-dev15~\cite{ignatiev2018pysat}.
The latter not only provides an interface to different SAT solvers, but also contains conversion routines to bring the encountered cardinality constraints~\eqref{eq:groupversion-sopeq-fulfillment-without-j-at-most-three},\eqref{eq:groupversion-explicit-three-subpermutation-fulfillment},\eqref{eq:groupversion-sopeq-fulfillment-without-j-at-most-three-right-coset},\eqref{eq:groupversion-explicit-three-subpermutation-fulfillment-right-coset} into pure CNF: Seven encodings are available and we pick the one recommended by the default strategy of \texttt{PySAT}.
For accessing the set of all (conjugacy classes of) subgroups of a given symmetric group $S_n$, we fall back on the \texttt{Julia}-package \texttt{Oscar} in version~1.2.0-dev~\cite{oscar2024computer}, having an interface to \texttt{GAP}~\cite{gap2024groups}.
In the literature, formulations for enforcing that for two binary vectors $a=(a_1,\ldots, a_r)$ and $b=(b_1,\ldots, b_r)$ the lexicographic ordering $a\preceq b$ applies (needed for the symmetry breaking in Section~\ref{sec:a-sat-approach}), have been proposed.
However, as these are unsupported by \texttt{PySAT}, we opted for the ``AND Encoding'' using ``Common Subexpression Elimination'' from~\cite{elgabou2015encoding}, which depends on fresh variables $x_i$, $i=1,\ldots,r-1$, and whose translation to a list of clauses reads as
\begin{align*}
  \neg x_1 \vee b_1 \vee \neg a_1,~ \neg x_1 \vee a_1 \vee \neg b_1,~ & x_1 \vee \neg a_1 \vee \neg b_1,~x_1 \vee a_1 \vee b_1, \\
  x_i \vee \neg x_{i+1},                                              & ~~~i=1,\ldots, r-2,                                     \\
  \neg x_{i+1} \vee b_{i+1} \vee \neg a_{i+1},                        & ~~~i=1,\ldots, r-2,                                     \\
  \neg x_{i+1} \vee a_{i+1} \vee \neg b_{i+1},                        & ~~~i=1,\ldots, r-2,                                     \\
  x_{i+1} \vee \neg b_{i+1} \vee \neg a_{i+1} \vee \neg x_i,          & ~~~i=1,\ldots, r-2,                                     \\
  x_{i+1} \vee b_{i+1} \vee a_{i+1} \vee \neg x_{i},                  & ~~~ i=1,\ldots, r-2,                                    \\
  \neg x_i \vee b_{i+1}\vee \neg a_{i+1},                             & ~~~i=1,\ldots, r-1.
\end{align*}
We chose for our experimental evaluation the SAT solver \texttt{Glucose} in version 4.2.1 with a time limit of 3600 seconds per instance if not stated otherwise.
The experiments were run on a cluster with an Intel(R) Xeon(R) E5-2640 v4 CPU with 2.40GHz and 160GB
RAM running Ubuntu 18.04.6 LTS on a single
thread.
Moreover, we decided not to re-examine the setting for $k=3$ whose properties due to Remarks~\ref{rem:identifiability-perfect-sequence-covering-arrays-rankwise-independent-families-and-further-literature}--\ref{rem:three-restricted-minwise-independence-is-three-rankwise-independence} can be directly read-off the works~\cite{gentle2023polynomial,na2023group} for $n\leq 8$.

\medskip

Table~\ref{tab:subgroups-results} reports obtained computational results as follows.
While the first two columns indicate the choice of parameters $d$, $n$, and $|G|$, the remainder contains results for the left- (shaded background) respectively right-coset approach.
Column \texttt{sg} gives information on how many conjugacy classes of subgroups of $S_n$ of order $|G|$ exist (\texttt{\#exst L}), respectively how many subgroups of $S_n$ of order $|G|$ exist (\texttt{\#exst R}).
The count of subgroups for which the experiment has been carried out can be read off the columns \texttt{\#cnsd L}/\texttt{R}.
The count of subgroups for which the left- respectively right-coset leads to a feasible solution is displayed in \texttt{sg-feas \#L}/\texttt{R}.
This information is flanked by the average time (in seconds), which the solver needed to find such a feasible solution (\texttt{avg\_t[s]} \texttt{L}/\texttt{R}).
The same structure applies for the group of columns \texttt{sg-infeas}, where these counts/measurements for subgroups leading to infeasibility are reported.
The right-coset approach is never run for $|G|=1$, symbolized by the entry ``/'', as this is the non-heuristic approach when just falling back on the left-coset approach.

\setlength{\tabcolsep}{1pt}
\begin{table}[htb!]
  \caption{Results for $k=4$; 3600s time limit given for each run of the SAT solver.
  Asterisks followed by numbers indicate the count of inconclusive experiments, i.e., the solver hit the time limit, for the right-coset approach.}\label{tab:subgroups-results}
      \centering
        \begin{tabular}{|b|b||a|b|a|b|a|b|a|b|a|b|a|b|}
          \hline
          ~          & ~     & \multicolumn{4}{c|}{\texttt{sg}}     & \multicolumn{4}{c|}{\texttt{sg-feas}} & \multicolumn{4}{c|}{\texttt{sg-infeas}}                                                                                                                                                                                                                                 \\
          ~          & ~     & \multicolumn{2}{c|}{\texttt{\#exst}} & \multicolumn{2}{c|}{\texttt{\#cnsd}}  & \multicolumn{2}{c|}{\texttt{\#}}        & \multicolumn{2}{c|}{\texttt{avg\_t[s]}} & \multicolumn{2}{c|}{\texttt{\#}}   & \multicolumn{2}{c|}{\texttt{avg\_t[s]}}\\
          $(d,n)$  & $|G|$ & L                                    & R                                     & L                                       & R                                       & L                                  & R                                       & L                                   & R                                         & L  & R    & L                                   & R                             \\\hline
          $(12,4)$ & 12 & 1 & 1 & 1 & 1 & 1 & 1 & $2.3\cdot 10^{-4}$ & $6.9\cdot 10^{-4}$ & 0 & 0 & / & / \\
           & 6 & 1 & 4 & 1 & 4 & 0 & 0 & / & / & 1 & 4 & $3.8\cdot 10^{-4}$ & $1.2\cdot 10^{-3}$ \\
           & 4 & 3 & 7 & 3 & 7 & 2 & 2 & $2.5\cdot 10^{-4}$ & $6.3\cdot 10^{-4}$ & 1 & 5 & $1.4\cdot 10^{-3}$ & $2.0\cdot 10^{-3}$ \\
           & 3 & 1 & 4 & 1 & 4 & 1 & 3 & $2.5\cdot 10^{-4}$ & $1.2\cdot 10^{-3}$ & 0 & 1 & / & $2.6\cdot 10^{-4}$ \\
           & 2 & 2 & 9 & 2 & 9 & 1 & 0 & $5.9\cdot 10^{-4}$ & / & 1 & 9 & $1.3\cdot 10^{-3}$ & $2.7\cdot 10^{-3}$ \\
           & 1 & 1 & / & 1 & / & 1 & / & $5.9\cdot 10^{-4}$ & / & 0 & / & / & / \\\hline
          $(12,5)$ & 12 & 2 & 15 & 2 & 15 & 0 & 0 & / & / & 2 & 15 & $7.6\cdot 10^{-4}$ & $1.9\cdot 10^{-3}$ \\
           & 6 & 3 & 30 & 3 & 30 & 1 & 0 & $1.1\cdot 10^{-3}$ & / & 2 & 30 & $2.2\cdot 10^{-3}$ & $2.6\cdot 10^{-2}$ \\
           & 4 & 3 & 35 & 3 & 35 & 0 & 0 & / & / & 3 & 35 & $8.7\cdot 10^{-3}$ & $2.1\cdot 10^{-1}$ \\
           & 3 & 1 & 10 & 1 & 10 & 1 & 0 & $8.4\cdot 10^{-4}$ & / & 0 & 10 & / & $4.5\cdot 10^{1\hphantom{-}}$ \\
           & 2 & 2 & 25 & 2 & 25 & 1 & 0 & $1.4\cdot 10^{-2}$ & / & 1 & 25 & $1.5\cdot 10^{-3}$ & $2.1\cdot 10^{-1}$ \\
           & 1 & 1 & / & 1 & / & 1 & / & $3.5\cdot 10^{-3}$ & / & 0 & / & / & / \\\hline
          $(12,6)$ & 1 & 1 & / & 1 & / & 0 & / & / & / & 1 & / & $4.1\cdot 10^{-2}$ & /  \\\hline
          $(24,6)$ & 24 & 6 & 90 & 6 & 90 & 2 & 9 & $2.6\cdot 10^{-3}$ & $8.0\cdot 10^{-3}$ & 4 & 81 & $1.1\cdot 10^{-2}$ & $1.6\cdot 10^{-2}$  \\
           & 12 & 4 & 150 & 4 & 150 & 2 & 12 & $5.4\cdot 10^{-3}$ & $2.5\cdot 10^{-1}$ & 2 & 138 & $3.2\cdot 10^{-3}$ & $4.7\cdot 10^{1\hphantom{-}}$ \\
           *2& 8 & 7 & 255 & 7 & 255 & 2 & 44 & $1.2\cdot 10^{-2}$ & $1.1\cdot 10^{-1}$ & 5 & 209 & $1.0\,\,\,\hphantom{\cdot 10^{-0}}$ & $1.8\cdot 10^{1\hphantom{-}}$ \\
           *7& 6 & 6 & 280 & 6 & 280 & 3 & 40 & $2.1\cdot 10^{-2}$ & $6.2\cdot 10^{-1}$ & 3 & 233 & $6.2\cdot 10^{-2}$ & $4.8\cdot 10^{1\hphantom{-}}$ \\
           *4& 4 & 7 & 255 & 7 & 255 & 5 & 76 & $4.7\cdot 10^{-2}$ & $1.8\,\,\,\hphantom{\cdot 10^{-0}}$ & 2 & 175 & $2.3\cdot 10^{1\hphantom{-}}$ & $9.8\cdot 10^{1\hphantom{-}}$ \\
           *9& 3 & 2 & 40 & 2 & 40 & 1 & 5 & $2.3\cdot 10^{-2}$ & $7.4\,\,\,\hphantom{\cdot 10^{-0}}$ & 1 & 26 & $6.3\cdot 10^{-1}$ & $5.3\cdot 10^{2\hphantom{-}}$ \\
           *18& 2 & 3 & 75 & 3 & 75 & 2 & 35 & $2.0\cdot 10^{-1}$ & $5.5\cdot 10^{2\hphantom{-}}$ & 1 & 22 & $3.9\,\,\,\hphantom{\cdot 10^{-0}}$ & $2.6\cdot 10^{2\hphantom{-}}$ \\
           & 1 & 1 & / & 1 & / & 1 & / & $6.7\cdot 10^{-1}$ & / & 0 & / & / & / \\\hline
          $(24,7)$ & 24 & 14 & 1435 & 14 & 1435 & 0 & 0 & / & / & 14 & 1435 & $2.9\cdot 10^{-2}$ & $8.0\cdot 10^{-2}$ \\
          *45& 12 & 13 & 1715 & 13 & 1715 & 2 & 0 & $4.2\cdot 10^{-2}$ & / & 11 & 1670 & $1.8\cdot 10^{-1}$ & $1.8\cdot 10^{1\hphantom{-}}$ \\
           *33& 8 & 7 & 1575 & 7 & 1575 & 0 & 0 & / & / & 7 & 1542 & $2.8\,\,\,\hphantom{\cdot 10^{-0}}$ & $3.4\cdot 10^{1\hphantom{-}}$ \\
           *354& 6 & 8 & 1645 & 8 & 1597 & 3 & 0 & $1.8\cdot 10^{-1}$ & / & 5 & 1243 & $5.2\cdot 10^{-1}$ & $2.1\cdot 10^{2\hphantom{-}}$ \\
           *23& 4 & 7 & 1295 & 7 & 104 & 3 & 0 & $7.8\cdot 10^{-1}$ & / & 4 & 81 & $1.0\cdot 10^{1\hphantom{-}}$ & $6.0\cdot 10^{2\hphantom{-}}$ \\
           *67& 3 & 2 & 175 & 2 & 81 & 1 & 0 & $7.0\,\,\,\hphantom{\cdot 10^{-0}}$ & / & 1 & 14 & $7.6\,\,\,\hphantom{\cdot 10^{-0}}$ & $8.0\cdot 10^{2\hphantom{-}}$ \\
           *63& 2 & 3 & 231 & 3 & 65 & 2 & 1 & $7.6\,\,\,\hphantom{\cdot 10^{-0}}$ & $3.2\cdot 10^{3\hphantom{-}}$ & 1 & 1 & $1.0\cdot 10^{3\hphantom{-}}$ & $1.8\cdot 10^{2\hphantom{-}}$ \\
           & 1 & 1 & / & 1 & / & 1 & / & $6.9\cdot 10^{1\hphantom{-}}$ & / & 0 & / & / & /\\\hline
        \end{tabular}
\end{table}

From the experiments in Table~\ref{tab:subgroups-results} we can derive the following information.
\begin{proposition}\label{existence-results-small-scale-four-to-seven-many-are-optimal}
  For $n\in\{4,5\}$, respectively $n\in\{6,7\}$, there are optimal $4$-restricted minwise independent families of $12$ members, respectively $24$ members, which are all coincident with a union of cosets of $S_n$.
  For $n=4$ respectively $n=6$ even a single coset of cardinality $12$ respectively $24$ constituting the family can be found.
\end{proposition}
\begin{proof}
  We have $\lcm{\natrange{4}}=12$.
  Therefore the representatives found in Table~\ref{tab:subgroups-results} are minimally-sized for $n\in\{4,5\}$.
  On the other hand, as according to entry $(12, 6, 4)$, $|G|=1$, in Table~\ref{tab:subgroups-results}, no minwise independent family $\mathcal{F}\subseteq S_6$ consisting of $12$ members exists, $d=24$ is the smallest next cardinality being a multiple of $12$ and for which the existence of such families should be explored.
  For $d=24$ and $n\in\{6,7\}$ such representatives, again being decomposable into cosets, indeed exist; see Table~\ref{tab:subgroups-results}.
\end{proof}
It is an interesting insight that consistently in Table~\ref{tab:subgroups-results} for any considered subgroup order, whenever non-decomposability into left-cosets occurs, also non-decomposability into right-cosets is observable.
On the other hand, for several subgroup orders, for which the left-coset approach was fruitful, the right-coset approach was in contrast not fruitful.
Moreover, the more complex SAT model for the right-coset approach considerably impacts the runtime by one to two orders of magnitude in comparison to the left-coset approach for feasible but also infeasible instances.
Due to these two observations we conclude that in particular for higher values of $d$ and $n$ we should rather focus on the left-coset approach to check more and apparently more promising instances within a fixed and limited margin of computation time.

\begin{proposition}\label{pro:nonexistence-four-restricted-minwise-independent-family--on-eight-symbols-on-twentyfour-members}
  There is no $4$-restricted minwise independent family $\mathcal{F}\subseteq S_8$ consisting of $24$ members.
\end{proposition}
\begin{proof}[computer-aided]
  The respective model turned out to be unsatisfiable after $64.7$ hours of computation time.
\end{proof}

Without time measurements and only reporting results for one fruitful subgroup of greatest cardinality, we add the following observation.
\begin{proposition}\label{pro:larger-d-and-n-values-experiments}
  $k$-restricted minwise independent families $\mathcal{F}\subseteq S_n$ with $d=|\mathcal{F}|$, which coincide with a union of left-cosets of $S_n$, exist for constellations
  \begin{equation*}
    (d,n,k;|G|) \in \{(60,6,6;6), (60,6,5;6), (60,5,5;60), (48,8,4;24)\}.
  \end{equation*}
\end{proposition}
\begin{proof}
  Certificates for this assertion are provided as supplementary material\footnote{\url{https://www.ac.tuwien.ac.at/files/resources/instances/minw-indep/supp.pdf}}.
\end{proof}
We also remark that for $k\in\{5,6\}$ we have $\lcm{\natrange{5}}=\lcm{\natrange{6}}=60$ and thus the respective instances in Proposition~\ref{pro:larger-d-and-n-values-experiments} are even optimal.
We highlight suboptimality for $(d,n,k)=(48,8,4)$, as we can construct an instance with constellation $(d,n,k)=(36,8,4)$ by applying Theorem~\ref{thm:bargachev-first-right-coset-approach-for-minwise-independence} to the $3$-restricted minwise independent family derivable from $\mathcal{F}\subseteq S_8$ of $18$ members in~\cite[Proposition~4.6]{na2023group}---due to Proposition~\ref{pro:nonexistence-four-restricted-minwise-independent-family--on-eight-symbols-on-twentyfour-members} the end product is even optimal.
However, without the latter trick the solver, even with runtimes up to $48$ hours, was not able to find any respective family counting $36$ members via decomposition into left-cosets.
Finally, we point out the particular situation occurring for $(d,k)=(12,4)$, see Table~\ref{tab:subgroups-results}, where for $n=4$ no decomposition into left-cosets of cardinality $6$ is possible, while this is the case for $n=5$, which intuitively is a stronger constraint setting.
This, apparently, is explainable by the considerably differing nature of the cardinality-$6$ subgroups of $S_5$.
At the same time we notice that for $|G|=24$ the left-coset approach worked for $n=4$ but not for $n=5$.

\section{Bijecting derangements onto non-waste permutations}\label{sec:bijecting}

This section settles the challenge of Bargachev~\cite[p.~6]{bargachev2004improved} to construct a bijection between two permutations classes exploited for obtaining the best known lower bound in Theorem~\ref{thm:asymptotic-overview}~\eqref{ite:bargachev-lower-bound}.
For this, the (later defined) class of permutations having $k$ \emph{waste indices}, denoted as $\wasteclass{n}{k}\subseteq S_n$, is introduced.
The pursued approach relies then on an algebraic proof of the coincidence of the cardinality of $\wasteclass{n}{k}$ with the count of \emph{$k$-partial derangements} in $S_n$, i.e., permutations with exactly $k\in\integerrange{0}{n}$ fixed points, whose class is henceforth denoted as $\kderang{n}{k}\subseteq S_n$.
However, Bargachev emphasizes the lack of a known explicit bijection $\Phi:\kderang{n}{k}\to \wasteclass{n}{k}$, even for $k=0$.

We fill this gap by more generally giving a bijection for \emph{arbitrary} $k$ having a remarkably succinct description.

\medskip

Let $\pi\in S_n$.
An index $j\in\natrange{n}$ is \emph{waste}~\cite{bargachev2004improved} if either the special case $j=n$ and $\min\{\pi(\ell): \ell=1,\ldots, n\} = \pi(n) = 1$ applies, or alternatively if we have $j<n$ with
\begin{equation}
  \min\{\pi(\ell): \ell=1,\ldots, j\} = \pi(j)\text{ and }\pi(j) > \pi(j+1).\label{eq:minimum-formulation-waste-index}
\end{equation}
Denote by $\wasteclass{n}{k}\subseteq S_n$ the set of all permutations possessing precisely $k\in [0:n]$ waste indices.

We now mimic and generalize the approach in~\cite[Proposition~2.1]{zhang2022combinatorics} relying on the cycle notation of permutations; see also~\cite{desarmenien1984autre}.
Retrospectively, this will highlight that analogues of the permutation classes introduced out of combinatorial interest in~\cite{zhang2022combinatorics} actually appear in more applied contexts.
Moreover, this gives a simpler alternative to the proof of~\cite[Lemma~4]{bargachev2004improved}, the latter requiring to solve a technical recurrence relation.
\begin{theorem}\label{thm:bijection-derangements-onto-waste-permutations}
  For $n\in\mathbb{N}$ and $k\in\integerrange{0}{n}$ there is an (explicitly describable) bijection $\Phi: \kderang{n}{k} \to \wasteclass{n}{k}$.
\end{theorem}
\begin{proof}
  Suppose $\pi\in\kderang{n}{k}\subseteq S_n$.
  Let us define the transformed $\Phi(\pi)$ as the result of the following procedure:
  Write down $\pi$ in cycle notation, where we assume that cycles  coincide with their cycle-representative placing the smallest entry on the cycle's first (i.e., lefter-most) position.
  Afterwards, the cycles are ordered from left to right such that their respective first entries decrease; see Example~\ref{exa:derangement-nive-four}.
  The final output $\Phi(\pi)$ is then defined as the permutation which maps $1,\ldots,n$ to the cycles' entries read from left to right.

  Now observe that the result $\Phi(\pi)$ is an element of $\wasteclass{n}{k}$:
  By construction, there will be precisely $k$ length-$1$ cycles in the cyclic notation of $\pi$.
  Each position of a length-$1$ cycle has an associated waste index:
  Note that the minimization in~\eqref{eq:minimum-formulation-waste-index} but also the the second, decreasing, behavior in~\eqref{eq:minimum-formulation-waste-index} holds true due to the ranked ordering of cycles.
  On the other hand, in cycles of length at least two, by construction, cycle-entries appearing not in the first position cannot be minimizers in~\eqref{eq:minimum-formulation-waste-index}.
  Moreover, the first entries of the cycles cannot meet the second condition of decrease in~\eqref{eq:minimum-formulation-waste-index}.
  No additional waste-indices originate from cycles of length at least two and therefore $\Phi(\pi)\in\wasteclass{n}{k}$.

  \medskip

  We show that $\Phi$ is surjective.
  For any $\tau\in\wasteclass{n}{k}$ we find $\pi\in\kderang{n}{k}$ such that $\Phi(\pi) = \tau$:
  Determine all indices $j_1<j_2< \ldots<j_\ell$ that satisfy the minimization-condition in~\eqref{eq:minimum-formulation-waste-index}.
  Using these indices, segment $\tau$ into $\ell$ blocks
  \begin{equation}
    (\tau_{j_1}, \tau_{j_1+1}, \ldots, \tau_{j_2-1}), (\tau_{j_2}, \tau_{j_2+1},\ldots, \tau_{j_3-1}), \ldots, (\tau_{j_\ell}, \tau_{j_\ell+1},\ldots, \tau_{n}).\label{eq:segmentation}
  \end{equation}
  Let $\pi$ be the permutation which results from interpreting the $\ell$ blocks in~\eqref{eq:segmentation} as cycles, i.e.,
  \begin{equation}
    \pi = (\tau_{j_1}\tau_{j_1+1}\cdots\tau_{j_2-1})(\tau_{j_2}\tau_{j_2+1}\cdots\tau_{j_3-1})\cdots(\tau_{j_\ell}\tau_{j_\ell+1}\cdots\tau_{n}).\label{eq:defined-via-blocks}
  \end{equation}
  As there are precisely $k$ non-waste indices for $\tau$, there are precisely $k$ length-$1$ blocks in~\eqref{eq:segmentation} and therefore $k$ length-$1$ cycles in~\eqref{eq:defined-via-blocks}.
  Eventually, this means that $\pi$ has exactly $k$ fixed points and that $\Phi(\pi)=\tau$.

  Next we show the injectivity of $\Phi$.
  The transition from a permutation to its decomposition into cycles is injective.
  For permutations from $\kderang{n}{k}$, this extends to the linear traversal of its cycles' members:
  In fact, the imposed decreasing ordering among the cycles' leaders and the minimum value attained by each leader in its cycle permit just one unique segmentation into cycles.
\end{proof}
\begin{example}\label{exa:derangement-nive-four}
  Let $(\rho(j))_{j=1}^9 = (1, 5, 3, 4, 6, 2, 8, 7, 9)$, and thus $\rho\in\kderang{9}{4}$.
  Then, its cycle notation, respectively $\Psi(\rho)$ constructed in Theorem~\ref{thm:bijection-derangements-onto-waste-permutations}, is given by
  \begin{equation*}
    \rho = (9)(78)(4)(3)(256)(1),\text{ respectively }\Phi(\rho) = (9,7,8,4,3,2,5,6,1).
  \end{equation*}
\end{example}
\begin{remark}
  Theorem~\ref{thm:bijection-derangements-onto-waste-permutations} automatically implies the validity of~\cite[Lemma 4]{bargachev2004improved} affirming that when the concept of waste indices is naturally lifted to subpermutations from $\subpermutations{n}{m}$ (compare~\cite[p.~3--4]{bargachev2004improved}), then for each fixed subset $X\subseteq \natrange{n}$ with $|X|=m$, we have $|\{\sigma\in\subpermutations{n}{m}:\sigma\text{ has codomain } X\text{ and no waste index}\}|={!m}$.
\end{remark}

\section{Conclusion}

We developed SAT models useful for examining the existence of (near-) optimal $k$-restricted minwise independent families $\mathcal{F}\subseteq S_n$ for reasonably small values of $n$ and $k$.
Besides finding many so-far unknown minwise independent families, our results show the following key observation concerning these structures:
Many optimal representatives of such families carry indeed the aforementioned group theoretic structure.
A further benefit of the decomposition is that it allows to find feasible representatives by avoiding prohibitively long running times, which in the setting of a purely non-heuristic approach, in contrast, would be needed.
We rounded our considerations up by describing a simple bijection leading to a short variant of a proof due to Bargachev~\cite{bargachev2004improved} that addresses rankwise independence.

The following aspects seem interesting for further work.
Firstly, examining the performance of an Integer Linear Programming approach, which could address the cardinality constraints straightforwardly, would be interesting, too.
Secondly, it would be interesting to see if instances so-far unsolved could be attacked by an incremental SAT approach.
Finally, given the large pool of subgroups for the right-coset approach, a further question is, if a solver can learn clauses from a small subpool of the subgroup-associated instances and then run faster on the remaining instances.

\subsubsection*{Acknowledgments}
  This research was funded in part by the program VGSCO of the Austrian Science Fund (FWF) [10.55776/W1260-N35].

\bibliographystyle{splncs04}

\begin{thebibliography}{10}
  \providecommand{\url}[1]{\texttt{#1}}
  \providecommand{\urlprefix}{URL }
  \providecommand{\doi}[1]{https://doi.org/#1}
  
  \bibitem{banbara2012generating}
  Banbara, M., Tamura, N., Inoue, K.: Generating event-sequence test cases by
  answer set programming with the incidence matrix. In: Technical
  Communications of the 28th International Conference on Logic Programming
  (ICLP'12). Schloss Dagstuhl-Leibniz-Zentrum f{\"u}r Informatik (2012)
  
  \bibitem{bargachev2006some}
  Bargachev, V.: On some properties of min-wise independent families and groups
  of permutations. Journal of Mathematical Sciences  \textbf{134}(5),
  2340--2345 (2006)
  
  \bibitem{bargachev2004improved}
  Bargachev, V.: An improved lower bound on the size of $k$-rankwise independent
  families of permutations (2004), preprint on webpage at
  \url{https://www.pdmi.ras.ru/preprint/2004/04-13.html}, last accessed on
  2024-12-06
  
  \bibitem{broder2000minwise}
  Broder, A.Z., Charikar, M., Frieze, A.M., Mitzenmacher, M.: Min-wise
  independent permutations. Journal of Computer and System Sciences
  \textbf{60}(3),  630--659 (2000)
  
  \bibitem{cameron2007minwise}
  Cameron, P.J., Spiga, P.: Min-wise independent families with respect to any
  linear order. Communications in Algebra  \textbf{35}(10),  3026--3033 (2007)
  
  \bibitem{oscar2024computer}
  Decker, W., Eder, C., Fieker, C., Horn, M., Joswig, M. (eds.): The {C}omputer
  {A}lgebra {S}ystem {OSCAR}: {A}lgorithms and {E}xamples, Algorithms and
  {C}omputation in {M}athematics, vol.~32. Springer, 1 edn. (8 2024)
  
  \bibitem{desarmenien1984autre}
  D{\'e}sarm{\'e}nien, J.: Une autre interpr{\'e}tation du nombre de
  d{\'e}rangements. S{\'e}minaire Lotharingien de Combinatoire  \textbf{8}, ~6
  (1982)
  
  \bibitem{elgabou2015encoding}
  Elgabou, H.: Encoding The Lexicographic Ordering Constraint in Satisfiability
  Modulo Theories. Ph.D. thesis, University of York (2015)
  
  \bibitem{gap2024groups}
  The GAP~Group: {GAP -- Groups, Algorithms, and Programming, Version 4.13.1}
  (2024), \url{https://www.gap-system.org}
  
  \bibitem{gentle2023polynomial}
  Gentle, A.R.: A polynomial construction of perfect sequence covering arrays.
  Algebraic Combinatorics  \textbf{6}(5),  1383--1394 (2023)
  
  \bibitem{gentle2023perfect}
  Gentle, A.R., Wanless, I.M.: On perfect sequence covering arrays. Annals of
  Combinatorics  \textbf{27}(3),  539--564 (2023)
  
  \bibitem{harvey2024explicit}
  Harvey, N., Sahami, A.: Explicit and near-optimal construction of $t$-rankwise
  independent permutations. In: Approximation, Randomization, and Combinatorial
  Optimization. Algorithms and Techniques (APPROX/RANDOM 2024). Schloss
  Dagstuhl--Leibniz-Zentrum f{\"u}r Informatik (2024)
  
  \bibitem{ignatiev2018pysat}
  Ignatiev, A., Morgado, A., Marques-Silva, J.: {PySAT}: {A} {P}ython toolkit for
  prototyping with {S}{A}{T} oracles. In: International Conference on Theory
  and Applications of Satisfiability Testing. pp. 428--437. Springer (2018)
  
  \bibitem{itoh2000permutations}
  Itoh, T., Takei, Y., Tarui, J.: On permutations with limited independence. In:
  Proceedings of the eleventh annual ACM-SIAM Symposium on Discrete Algorithms.
  pp. 137--146 (2000)
  
  \bibitem{iurlano2023growth}
  Iurlano, E.: Growth of the perfect sequence covering array number. Designs,
  Codes and Cryptography  \textbf{91}(4),  1487--1494 (2023)
  
  \bibitem{kuperberg2017probabilistic}
  Kuperberg, G., Lovett, S., Peled, R.: Probabilistic existence of regular
  combinatorial structures. Geometric and Functional Analysis  \textbf{27}(4),
  919--972 (2017)
  
  \bibitem{levenshtein1992perfect}
  Levenshtein, V.I.: On perfect codes in deletion and insertion metric. Discrete
  Mathematics and Applications  \textbf{2}(3),  241--258 (1992)
  
  \bibitem{li2011theory}
  Li, P., K{\"o}nig, A.C.: Theory and applications of $b$-bit minwise hashing.
  Communications of the ACM  \textbf{54}(8),  101--109 (2011)
  
  \bibitem{mathon1999directed}
  Mathon, R., {v}an Trung, T.: Directed $t$-packings and directed $t$-{S}teiner
  systems. Designs, Codes and Cryptography  \textbf{18}(1),  187--198 (1999)
  
  \bibitem{mulmuley1994computational}
  Mulmuley, K.: Computational geometry : an introduction through randomized
  algorithms. Prentice-Hall, Englewood Cliffs, N.J (1994)
  
  \bibitem{na2023group}
  Na, J., Jedwab, J., Li, S.: A group-based structure for perfect sequence
  covering arrays. Designs, Codes and Cryptography  \textbf{91}(3),  951--970
  (2023)
  
  \bibitem{oeis2024online}
  {OEIS Foundation Inc.}: The {O}n-{L}ine {E}ncyclopedia of {I}nteger {S}equences
  (2024), published electronically at \url{https://oeis.org}
  
  \bibitem{peczarski2004new}
  Peczarski, M.: New results in minimum-comparison sorting. Algorithmica
  \textbf{40},  133--145 (2004)
  
  \bibitem{tarui2003nearly}
  Tarui, J., Itoh, T., Takei, Y.: A nearly linear size 4-min-wise independent
  permutation family by finite geometries. In: Arora, S., Jansen, K., Rolim,
  J.D.P., Sahai, A. (eds.) Approximation, Randomization, and Combinatorial
  Optimization. Algorithms and Techniques. pp. 396--408. Springer, Berlin,
  Heidelberg (2003)
  
  \bibitem{yuster2020perfect}
  Yuster, R.: Perfect sequence covering arrays. Designs, Codes and Cryptography
  \textbf{88}(3),  585--593 (2020)
  
  \bibitem{zamora2016hashing}
  Zamora, J., Mendoza, M., Allende, H.: Hashing-based clustering in high
  dimensional data. Expert Systems with Applications  \textbf{62},  202--211
  (2016)
  
  \bibitem{zhang2022combinatorics}
  Zhang, J., Gray, D., Wang, H., Zhang, X.D.: On the combinatorics of
  derangements and related permutations. Applied Mathematics and Computation
  \textbf{431},  127341 (2022)
  
\end{thebibliography}
\end{document}